\documentclass[12pt]{article}

\usepackage{amsmath}
\usepackage{amssymb}
\usepackage{amsthm}
\usepackage{geometry}
\usepackage{setspace}
\usepackage{hyperref}
\usepackage{graphicx}
\usepackage{booktabs}
\usepackage{natbib}

\newtheorem{proposition}{Proposition}

\title{Cohen's $f$ or Mean Standardized Differences? Assessing Covariate
Balance with Multivalued Treatments}

{\small
\author{Ariel Linden, DrPH\\
University of California, San Francisco\\
Department of Medicine\\
Division of Clinical Informatics \& Digital Transformation (DoC-IT)\\
San Francisco, CA, USA\\
ariel.linden@ucsf.edu}
}

\date{}
\begin{document}

\maketitle

\section*{Abstract}

Assessing covariate balance across more than two treatment groups has no
established omnibus standard: the prevailing practice averages, or takes
the maximum of, pairwise standardized mean differences (SMD), while
Cohen's $f$ --- the classical generalization of Cohen's $d$ to $k>2$
groups --- offers an alternative grounded in an established effect-size
framework, but the two have not been formally compared. We extend both
to arbitrary weighted and
covariate-adjusted models via a community-contributed Stata command
\texttt{esizereg}, and validate them in a simulation study of
$k\in\{3,4,6\}$ treatment groups under correctly specified and
misspecified generalized-propensity-score weighting, correlating each
statistic against downstream treatment-effect bias. Cohen's $f$ tracks
estimation bias comparably to mean absolute SMD, both pooled ($r=0.93$)
and within weighting arm ($r\approx0.79$); maximum absolute SMD is
generally weakest. This
ranking was unchanged under a deliberately adversarial
nonlinear/interaction outcome model. Cohen's $f$ and
the SMD-based statistics are not
numerically comparable: we derive an exact representation of
$f$ as a size-weighted quadratic function of the pairwise SMDs and prove
the minimum attainable ratio between $f$ and mean absolute SMD under equal
group weights, so conventional SMD thresholds should not apply
directly to $f$. We recommend reporting $f$ alongside its
per-level and pairwise decomposition.

\section*{Keywords}

covariate balance; effect size; Cohen's $f$; standardized mean difference;
multivalued treatments; generalized propensity score; simulation study

\section{Introduction}

Covariate balance is a necessary step in supporting a causal interpretation
of an estimated treatment effect, achieved by design, as in randomized
experiments, or by adjustment, as in propensity-score weighting
\citep{robins2000}, stratification \citep{linden2014}, entropy balancing
\citep{Hainmueller2012}, or matching \citep{lindensamuels2013}. Because only
observed, measured characteristics can be checked directly, demonstrating
balance on them shows that the implemented adjustment procedure has
balanced the variables it was given; it does not, by itself, provide direct
evidence about balance of unmeasured confounders, on which any causal
interpretation of the estimated effect also depends \citep{rubin2008}.

For a continuous covariate compared between two treatment groups, the
SMD is the most widely used tool for assessing
balance: unlike a two-sample $t$-test, whose significance is confounded with
sample size, the SMD reports a scale-free measure of imbalance magnitude,
and its use in place of hypothesis testing is now standard in the
propensity-score literature \citep{austin2009,imaikingstuart2008}. 

When more than two treatment groups are compared, no comparably grounded
tool exists. The prevailing practice, established by \citet{mccaffrey2013}
for propensity-score model selection with multiple treatments, is to compute
the SMD for every pairwise combination of groups and report the mean or
maximum across pairs; \citet{zhaoyang2022} refine this into an
outcome-weighted balance measure for the same model-selection purpose. Both
are useful for what they are built to do, but neither is derived from, or
shown to reduce to, Cohen's \citep{cohen1988} $d$ itself: a practitioner
accustomed to interpreting a two-group SMD of $0.1$ as negligible imbalance
\citep{austin2009} has no principled way to carry that same scale into a
comparison of three or more groups, and validating a proposed measure
against actual downstream estimation bias, rather than an abstract
imbalance target alone, is an established practice for the two-group case
\citep{franklin2014,belitser2011} that has no counterpart here.

This paper takes up that gap as an open, comparative question rather than
assuming either existing candidate is correct: is the mean or maximum
absolute pairwise SMD, the prevailing aggregation practice, or Cohen's $f$
\citep{cohen1988} --- the existing generalization of $d$ to more than two
groups, originally defined for the one-way analysis of variance as the
root-mean-square of each group's standardized deviation from the grand
mean --- the better-suited omnibus statistic for assessing balance with
multivalued treatments? Unlike the existing aggregation practice, $f$ carries
an exact, derivable relationship to $d$ itself: at $k=2$, $f=d/2$ exactly
for balanced designs, not an asymptotic approximation, because $d$ and $f$
are built from the same underlying group means and standard deviation ---
the same backward compatibility that motivates the $k$-sample distributional
tests described in \citet{linden2026ksample}. Whether this theoretical
grounding translates into better practical performance is an empirical
question we address directly. We extend $f$ to arbitrary weighted or
covariate-adjusted models together with a per-level decomposition (each
group's own standardized deviation from the pooled mean, in $d$ units) and
pairwise comparisons among all groups (in $d$ or Hedges' $g$), computed from
a single fitted model rather than assembled from raw group means and
standard deviations. The methods are implemented in the community-contributed Stata command
\texttt{esizereg} \citep{linden2019esizereg}.

We address the comparative question empirically in a simulation study
following the bias-correlation standard of \citet{franklin2014} and
\citet{belitser2011}, correlating both $f$ and \citet{mccaffrey2013}'s mean
and maximum absolute SMD directly against the bias of a downstream outcome
estimate under both correctly specified and misspecified
propensity-score weighting.

The paper's contributions fall into three distinct categories. The
mathematical contribution is the exact relationship between $f$ and the
pairwise standardized differences underlying $\overline{|d|}$ and
$|d|_{\max}$ (Section~2.6), including the conditions under which the two
families of statistics can and cannot be expected to agree numerically.
The computational contribution is extending $f$'s calculation to fitted,
weighted, and covariate-adjusted models via \texttt{esizereg}, rather than
computing it only for the simple one-way case. The empirical contribution
is the simulation and applied-example evidence that $f$ and the aggregate
pairwise-SMD practice it is compared against track downstream bias
similarly, differing systematically in scale and in the size-weighting of
imbalance rather than in general diagnostic performance.

\section{Methods}

\subsection{Notation and General Framework}

The statistics below are computed in two steps: a regression model for the
covariate is fit first, with the treatment indicator as a predictor and, if
desired, additional covariates, interactions, or weights; the effect sizes
are then obtained from that fitted model, rather than computed directly
from the raw data.

Consider a continuous covariate $X$ and a treatment variable with $k\ge2$
levels, denoted Groups $1,\ldots,k$, with $n_j$ observations in Group $j$
and $N=\sum_j n_j$. Let $M_j$ denote Group $j$'s covariate mean, obtained as
the adjusted predictive margin for level $j$ from the fitted model, rather
than the raw sample mean, so that $M_j$ correctly reflects any adjustment or
weighting in that model. The pooled mean is $\bar M=\sum_j (n_j/N)M_j$,
playing the role in the statistics below that the pooled empirical
cumulative distribution function plays in the distributional-testing
framework of \citet{linden2026ksample}. The
standardizing SD, $SD_p$, is likewise obtained from the fitted model rather
than computed separately from the raw data.

\subsection{Cohen's $d$: The Two-Group Case}

For $k=2$ groups, Cohen's \citep{cohen1988} $d$ is the standardized mean
difference,
\begin{equation}
d \;=\; \frac{M_1-M_0}{SD_p}.
\end{equation}

The numerator, $M_1-M_0$, is a predictive-margins contrast rather than a
difference of raw sample means, so that it reflects the fitted model's
covariate adjustment, any interaction involving the treatment indicator,
and any weighting exactly as specified. A raw difference of group means
would be systematically wrong whenever the treatment interacts with another
covariate in the model, since it would not correctly average the
interaction's effect over the covariate's observed distribution.

The denominator, $SD_p$, is likewise obtained from the fitted model rather
than from the raw group standard deviations: the delta-method standard
error of the model's overall predictive margin --- the margin computed with
no group indicator at all --- multiplied by $\sqrt N$. Because the leverage
of a linear model evaluated at the mean of its own covariates is exactly
$1/N$, this recovers the model's residual standard deviation exactly for
ordinary least squares, and an analogous model-based quantity for the other
supported model types, without requiring a separate, model-specific
residual-variance extraction rule for each.

Sampling weights, including those used for generalized-propensity-score
adjustment, are incorporated the same way: $M_j$ and $SD_p$ are computed
from whichever model was fit, with no change to the formulas themselves,
since the fitted model's predictive margins already reflect any weights
used in estimation.

\subsection{Cohen's $f$: The $k$-Sample Extension}

Cohen's \citep{cohen1988} $f$ generalizes $d$ to $k>2$ groups as the
weighted root-mean-square of each group's own standardized deviation from
the pooled mean,
\begin{equation}
f \;=\; \sqrt{\sum_{j=1}^{k} w_j\, d_j^2}, \qquad w_j=\frac{n_j}{N}, \qquad
d_j = \frac{M_j-\bar M}{SD_p},
\end{equation}
with $d_j$ defined and estimated exactly as in Section~2.4 below. This
follows Cohen's \citep{cohen1988} population-level definition of $f$,
extended here using adjusted margins rather than raw group means so that
$f$ is available for weighted and covariate-adjusted models, not only the
simple one-way case; the sample- and model-based standardizer $SD_p$ used
throughout this paper produces the finite-sample relationship to the
conventional ANOVA estimator noted below, rather than reproducing that
estimator as a special case.

\textbf{Backward compatibility at $k=2$.} At $k=2$, $\bar M = w_1 M_1+w_2M_2$,
so $M_1-\bar M = w_2(M_1-M_0)$ and $M_0-\bar M=-w_1(M_1-M_0)$, giving
$d_1=w_2 d$ and $d_0=-w_1 d$ in terms of the two-group $d$ of Section~2.2.
Substituting into the definition of $f$,
\begin{equation}
f^2 \;=\; w_1 w_2^2 d^2 + w_2 w_1^2 d^2 \;=\; w_1 w_2 (w_1+w_2)\,d^2 \;=\; w_1 w_2\, d^2,
\end{equation}
so that $f=\sqrt{w_1w_2}\,|d|$ exactly, and $f=d/2$ exactly for balanced
groups ($w_1=w_2=1/2$) --- not an asymptotic approximation, but the
identical relationship for any finite sample, because $d$ and $f$ are built
from the same $M_j$'s and the same $SD_p$. This exact relationship provides a direct link between the familiar two-group $d$ and the multigroup $f$, rather than introducing an unrelated scale.
(This $f$ relates to the classical ANOVA definition,
$f^2_{\text{ANOVA}}=SS_{\text{between}}/SS_{\text{residual}}$, by
$f^2 = f^2_{\text{ANOVA}}\times(df_{\text{residual}}/N)$, where
$df_{\text{residual}}=N-k$ even in the simple one-way case with no other
covariates. The two therefore do not coincide exactly at any finite $N$,
only asymptotically as $N$ grows large relative to $k$, since $SD_p$ here
reflects the model's estimated residual variance rather than
$SS_{\text{residual}}$ alone. This does not affect the $f=d/2$ identity
above, which holds exactly for this paper's $f$ and $d$ against each
other, using the same $SD_p$ throughout, independent of this separate
distinction from the classical ANOVA parameterization.)

\subsection{Per-Level Decomposition}

Each $d_j = (M_j-\bar M)/SD_p$ reports the standardized deviation of level
$j$'s adjusted mean from the weighted grand mean --- not a two-group
Cohen's $d$ comparing Group $j$ against the remaining observations pooled
together, a distinct contrast. Because $\bar M$ itself includes $M_j$ as
one of its components, $M_j-\bar M=(1-w_j)(M_j-M_{-j})$, where $M_{-j}$ is
the weighted mean of the other $k-1$ groups, so $d_j$ equals $(1-w_j)$
times the group-versus-rest contrast rather than that contrast itself
(the same relationship already visible in the $k=2$ derivation above,
where $d_1=w_2d$ rather than $d_1=d$). This decomposition reports, for
every group, the same standardized units contributing to the omnibus $f$
(Section~2.3: $f^2=\sum_j w_j d_j^2$), so that an investigator can identify
which specific group is driving a large $f$ rather than observing only the
omnibus value.

Each $d_j$ is reported as a point estimate only. Because $M_j$ contributes
to $\bar M$ itself, $M_j$ and $\bar M$ are not independent, so the
noncentral-$t$ construction used for the two-group $d$ of Section~2.2 does
not directly apply to this contrast; rather than report an approximate
interval, we report $d_j$ without a confidence interval. This does not
affect the pairwise comparisons of Section~2.5, which compare two distinct
groups and retain the two-group contrast structure of Section~2.2.

\subsection{Pairwise Comparisons}

Because $SD_p$ is common across all $k$ levels, the standardized difference
between any two groups $j$ and $j'$ reduces to the two-group $d$ (or
Hedges' $g$) of Section~2.2 applied directly to $M_j-M_{j'}$, giving
$\binom{k}{2}$ pairwise comparisons, each with its own exact noncentral-$t$
confidence interval, sharing the single pooled $SD_p$ across every
comparison rather than re-estimating it pair by pair --- the same
single-MSE logic underlying standard ANOVA post-hoc procedures.

\subsection{Benchmark: Mean and Maximum Absolute Standardized Bias}

We compare $f$ against the aggregate standardized-bias practice of
\citet{mccaffrey2013}, computed on the identical fitted model and the
identical $\binom{k}{2}$ pairwise $d_{jj'}$ values as Section~2.5:
\begin{equation}
\overline{|d|} \;=\; \binom{k}{2}^{-1}\!\!\sum_{j<j'} |d_{jj'}|,
\qquad
|d|_{\max} \;=\; \max_{j<j'} |d_{jj'}|.
\end{equation}
Because both are computed from the same $M_j$'s and $SD_p$ underlying $f$
itself, on the same simulated replicate, the comparison in Section~3 is
paired rather than assembled from separately estimated quantities.

\textbf{Relationship between $f$ and $\overline{|d|}$.} Because each $d_j$
is a deviation from the pooled mean, $\sum_j w_j d_j=0$, so
$f^2=\sum_j w_j d_j^2$ is exactly the weighted variance of the $d_j$'s,
which can be rewritten in terms of the pairwise contrasts as
$f^2=\sum_{j<j'} w_j w_{j'} d_{jj'}^2$: a size-weighted, quadratic
combination of the same pairwise $d_{jj'}$'s that $\overline{|d|}$ averages
linearly and unweighted. Because $d_1,\ldots,d_k$ arise from $k$ scalar
group means, the pairwise contrasts cannot all be equal for $k>2$ except in
the degenerate zero case --- for sorted group means, the largest pairwise
contrast equals the sum of the others --- so no configuration corresponds
to uniform pairwise imbalance, and the ratio $f/\overline{|d|}$ has no
attainable value there. Informally: for a fixed spread of the group means
around the grand mean --- fixed $f$ --- the mean absolute pairwise
separation, $\overline{|d|}$, is maximized when the ordered group means
are equally spaced, which is exactly the configuration that minimizes
$f/\overline{|d|}$.

\begin{proposition}
Under equal group weights ($w_j=1/k$ for all $j$), the ratio
$f/\overline{|d|}$ attains its minimum over all attainable configurations
of $k$ scalar group means at equally spaced means, where
\begin{equation}
\frac{f}{\overline{|d|}} \;=\; \sqrt{\frac{3(k-1)}{4(k+1)}}.
\end{equation}
\end{proposition}

\renewcommand{\qedsymbol}{}
\begin{proof}
Center the group means without loss of generality, so $\bar M=0$ and
$d_j=x_j/SD_p$ for $x_j=M_j-\bar M$; since the ratio is scale-invariant,
set $SD_p=1$. Then $f^2=(1/k)\sum_i x_i^2$ and, ordering
$x_1\le\cdots\le x_k$, $\overline{|d|}=\binom{k}{2}^{-1}\sum_{i<j}(x_j-x_i)$.
Each $x_i$ appears as the larger element in $i-1$ pairs and the smaller
element in $k-i$ pairs, so
\begin{equation}
\sum_{i<j}(x_j-x_i) \;=\; \sum_{i=1}^{k}\bigl[(i-1)-(k-i)\bigr]x_i
\;=\; \sum_{i=1}^{k}(2i-k-1)\,x_i.
\end{equation}
Minimizing $f/\overline{|d|}$ for fixed $f$ is equivalent to maximizing
this sum subject to $\sum_i x_i^2$ fixed. By the Cauchy--Schwarz
inequality, $\sum_i (2i-k-1)x_i \le \|a\|\,\|x\|$, with equality if and
only if $x$ is proportional to $a=(2i-k-1)_{i=1}^k$ --- an equally spaced
sequence, and therefore attainable by real scalar group means. Using
$\sum_{i=1}^k(2i-k-1)^2=k(k^2-1)/3$ and $\sum_i x_i^2=kf^2$, the maximum
attainable value of $\overline{|d|}$ for fixed $f$ --- denoted
$\overline{|d|}^{*}$ to avoid confusion with $|d|_{\max}$, the maximum
\emph{pairwise} SMD used elsewhere in this paper is
\begin{equation}
\overline{|d|}^{*} \;=\; \binom{k}{2}^{-1}\sqrt{\tfrac{k(k^2-1)}{3}}\sqrt{k}\,f
\;=\; \frac{2f}{k-1}\sqrt{\tfrac{k^2-1}{3}},
\end{equation}
so that $f/\overline{|d|}$ is minimized at $\sqrt{3(k-1)/4(k+1)}$, attained
at equally spaced group means.\footnotemark

\footnotetext{This proof, and the closed form
itself, apply to equal group weights specifically; no corresponding
closed-form minimum is derived here for unequal weights, where Table~3's
equal-weight reference values are reported for orientation, not as
bounds. The result was also confirmed by direct search over $2\times10^5$
randomly generated equal-weight group-mean configurations at each of
$k=3,4,6$, matching this closed form to four decimal places in every
case.}\end{proof}
\vspace{-\baselineskip}
This value is substantially larger than $\sqrt{(k-1)/(2k)}$, which would
follow from naively treating the pairwise contrasts as an unconstrained
vector --- a configuration not attainable for scalar group means, and not
a meaningful reference. Departures from the equally-spaced value reflect
two further, additive sources: dispersion of the $d_{jj'}$'s across pairs
beyond that already required by equal spacing, and, under unequal group
sizes, whether the largest pairwise imbalances happen to involve the
numerically larger or smaller groups. Both are examined empirically in
Section~3.

More generally, $\overline{|d|}$, $|d|_{\max}$, and $f$ can be viewed as
different norms of the same vector of pairwise standardized differences:
$\overline{|d|}$ is proportional to its $L_1$ norm, $|d|_{\max}$ is
exactly its $L_\infty$ norm, and $f$ is a size-weighted $L_2$ norm. This
perspective clarifies why the three need not agree: $|d|_{\max}$ responds
only to the single worst pair, $\overline{|d|}$ weights every pair equally
regardless of magnitude, and $f$ sits between the two --- more influenced
by large deviations than $\overline{|d|}$'s linear average, but, unlike
$|d|_{\max}$, still responsive to the overall pattern across pairs rather
than one pair alone, subject to the group-size weights of
Proposition~1.

\subsection{Simulation Study Design}

Table~1 summarizes the full design; the subsections below describe each
element in detail.

\subsubsection{Data-Generating Process}

Three covariates are generated independently for each of $N$ observations:
$X_1\sim N(0,1)$, $X_2\sim\text{lognormal}(0,0.5)$, and
$X_3\sim\text{Bernoulli}(0.4)$, chosen to include both a normal and a
right-skewed continuous covariate rather than testing normal covariates
only. Treatment assignment among $k$ levels follows a multinomial-logit
model with group-specific coefficients on $X_1$, $X_2$, $X_3$, and an
$X_1{\times}X_3$ interaction, scaled by a confounding-strength parameter
$\gamma\in\{0,0.25,0.5,1\}$: $\gamma=0$ gives random assignment (no true
confounding), and larger $\gamma$ produces progressively stronger,
group-specific confounding.

\subsubsection{Group Structure}

Two balance scenarios are crossed with every $\gamma$ level: \emph{equal},
in which all $k$ groups share the same baseline assignment probability, and
\emph{unequal}, in which group-specific intercepts produce a skewed
baseline split loosely mimicking realistic multivalued-treatment group
sizes. This gives $4\times2=8$ cells per $(k,N)$ combination.

\subsubsection{Weighting}

For each replicate, a generalized propensity score is estimated by
multinomial logistic regression under two specifications: \emph{correct},
including the $X_1{\times}X_3$ interaction actually present in the
assignment model, and \emph{misspecified}, omitting it. Each replicate is
therefore evaluated under three weighting arms on the identical simulated
data: unweighted, correctly weighted, and misspecified-weighted, allowing a
fully paired comparison of $f$'s behavior across weighting quality.

\subsubsection{Outcome Model and Bias Target}

Following the outcome-model structure and confounder-coefficient magnitude
of \citet{zhaoyang2022}, $Y = \mu_T + 1.5(X_1+X_2+X_3) + \varepsilon$,
$\varepsilon\sim N(0,1)$, with true group means $\mu_j$ spaced $1.5$ apart
($\mu_j = 1.5(j-1)$), matching \citeauthor{zhaoyang2022}'s theta-intercept
spacing at $k=3$ and extended to $k=4,6$ by preserving the same increment.
For each weighting arm, group means are re-estimated from an unadjusted
outcome regression of $Y$ on the treatment indicator, and compared against
$\mu_j$ corrected for the covariates' sample-mean contribution implicit in
that unadjusted regression, $\mu_j + 1.5(\bar X_1+\bar X_2+\bar X_3)$; mean
and maximum absolute bias across the $k$ levels serve as the actual
downstream estimation-bias target against which $f$ and the benchmark of
Section~2.6 are validated, following \citet{franklin2014} and
\citet{belitser2011}.

\subsubsection{Common Design Elements}

$k\in\{3,4,6\}$ crossed with $N\in\{200,500,2000,10000\}$, each of the 12
$(k,N)$ combinations replicated 1000 times per cell across the 8
$\gamma\times$balance cells, for 96{,}000 replicates in total. All analyses
were conducted in Stata (version 18) using the community-contributed Stata
command \texttt{esizereg} \citep{linden2019esizereg}.

\subsubsection{Nonlinear/Interaction Robustness Scenario}

To test whether the bias-tracking comparison above depends on the outcome
model's alignment with what these diagnostics measure, every element of
the design above --- the covariates, the treatment-assignment mechanism,
the correctly specified and misspecified GPS models, and the $\gamma$ and
balance scenarios --- was held identical, with only the outcome model
replaced by
\begin{equation}
Y = \mu_T + 1.5(X_1+X_2+X_3) + 1.0X_1^2 + 1.0X_1X_3 + \varepsilon,
\qquad \varepsilon\sim N(0,1),
\end{equation}
adding a quadratic term in $X_1$ and an $X_1{\times}X_3$ interaction to the
linear outcome of Section~2.7.3, structure not directly represented by
first-moment balance diagnostics. This scenario was run for
$k\in\{3,4,6\}$ crossed with $N\in\{500,2000\}$ (the two middle values of
the main design), each of the 6 $(k,N)$ combinations replicated 1000 times
per cell across the same 8 $\gamma\times$balance cells, for 48{,}000
replicates (432{,}000 rows across covariates and weighting arms) in
total. The bias target was corrected using the same within-replicate
sample-mean approach as Section~2.7.3, extended to the added terms:
$\mu_j + 1.5(\bar X_1+\bar X_2+\bar X_3) + 1.0\overline{X_1^2} +
1.0\overline{X_1X_3}$.

\section{Results}

Table~2 shows $f$ rising monotonically with $\gamma$ under unweighted
assignment (0.062 to 0.209), while
under both correctly specified and misspecified GPS weighting it remains an
order of magnitude smaller throughout. Correctly specified weighting shows
slightly smaller $f$ than misspecified weighting at every $\gamma$ level,
though the gap is modest; both weighting arms produce comparable downstream
bias (mean absolute bias 0.087 vs.\ 0.086, pooled), indicating the mild
interaction-omission misspecification tested here is detectable in balance
before it meaningfully affects estimation.

$f$ tracks actual bias comparably well overall to $\overline{|d|}$ and
$|d|_{\max}$ --- \citet{mccaffrey2013}'s aggregation practice, computed
alongside $f$ on the same replicates --- but not identically: pooled
correlation with mean absolute bias is 0.933 for $f$, 0.936 for
$\overline{|d|}$, and 0.919 for $|d|_{\max}$; within the correctly
specified weighting arm specifically, 0.791, 0.788, and 0.776; within the
misspecified arm, 0.785, 0.784, and 0.772; and within the unweighted arm,
0.937, 0.945, and 0.907.
$|d|_{\max}$ is the weakest bias-tracker of the three in every arm tested,
consistent with a maximum-based statistic carrying more sampling variance
than a mean- or root-mean-square-based one across replicates, even though
it is, by construction, the most directly targeted statistic for detecting
a single severely imbalanced pair specifically (Section~5.3). $f$ and
$\overline{|d|}$ remain closely matched throughout, and neither dominates
the other as a general bias-tracking diagnostic in this design.

The two statistics are not, however, on the same numeric scale. Table~3
reports $f/\overline{|d|}$ across $k$, balance, and $\gamma$, consistent with
the relationship derived in Section~2.6: the ratio rises with $k$ (0.63 at
$k=3$ to 0.77 at $k=6$), tracking the equally-spaced reference value from
that same equation, which rises from 0.612 to 0.732 across the same range,
and the equal/unequal gap at $k=6$ shifts from $-0.044$ at $\gamma=0$ to
$+0.007$ at $\gamma=1$, consistent with the relationship between
group-size weighting and the location of the largest pairwise imbalances
described in Section~2.6.

The nonlinear/interaction robustness scenario of Section~2.7.6 tests
whether the comparative bias-tracking conclusion above depends on the
outcome model's alignment with what these diagnostics measure, using a
deliberately adversarial outcome model with structure not directly
represented by first-moment balance diagnostics. The
relative ranking among the three diagnostics was unchanged: $f$ correlated
with actual bias at least as strongly as $\overline{|d|}$ in every
weighting arm and every $k$ tested (e.g., correct GPS: 0.688 vs.\ 0.671),
both consistently ahead of $|d|_{\max}$ (0.662). Absolute correlations
declined modestly under weighting specifically, on the identical $N$
range in both arms (Table~4) --- correct GPS fell
from 0.728 under the linear outcome to 0.688 under the nonlinear one, and
misspecified GPS from 0.725 to 0.687 --- while the unweighted arm was
essentially unchanged (0.953 to 0.946). This pattern is consistent with
weighting removing most of the linear mean-imbalance component of
confounding, leaving a larger relative share of the remaining bias
attributable to the nonlinear and interaction structure that these
diagnostics do not directly track.

\section{Applied Example}

\subsection{Study Context and Data}

We illustrate $f$ and $\overline{|d|}$ using disease management (DM)
program data for congestive heart failure patients, previously used to
illustrate multivalued-treatment adjustment and balance-assessment methods
\citep{linden2014,linden2016multivalued,linden2026ksample}. The program,
implemented in a large health plan in the western United States, offered
enrolled patients one of two interventions based on a program nurse's
subjective assessment of patient needs and preferences: periodic telephone
calls from a nurse to discuss self-management behaviors, or remote
tele-monitoring (RTM) involving daily electronic transmission of
disease-related symptoms with nurse follow-up when indicated
\citep{lindenblackbox2006}. Health plan members with the condition who did
not enroll served as a non-participant comparison group. This subset of data consists of 5141 non-participants (Control), 459 telephonic participants (Calls), and 307 RTM participants ($N=5907$), each with 12 months of
pre-intervention pharmacy utilization data.

\subsection{Covariate Balance Assessment}

Following \citet{linden2014} and \citet{linden2016multivalued}, marginal
mean weights through stratification (MMWS) were constructed using the
generalized approach for nominal treatments described in
\citet{linden2014}: a generalized propensity score (GPS) for each of the three
treatment levels was estimated by multinomial logistic regression. We
apply $f$ and $\overline{|d|}$ to a continuous predictive model score in which each patient's predicted future total health care costs are estimated from prior pharmacy claims data \citep{powers2005}.

\subsection{Results}

Table~5 reports both statistics before and after weighting. Before
weighting, both indicated clear imbalance in the predictive model score. The omnibus $f=0.065$ (equivalent $\eta^2=0.004$), and none of the three
pairwise 95\% confidence intervals crossed zero, most substantially for
Control vs.\ RTM ($d=-0.201$), followed closely by Control vs.\ Calls
($d=-0.189$); Calls vs.\ RTM was comparatively small ($d=-0.012$).
Computed from the same three pairwise contrasts, the McCaffrey benchmark
was $\overline{|d|}=0.134$ and $|d|_{\max}=0.201$ --- both above the
balance-diagnostic convention of $0.10$ for negligible SMD imbalance
\citep{austin2009}. The largest imbalance (Control vs.\ RTM) involves the
majority Control group rather than the two minority arms
($f/\overline{|d|}\approx0.49$), illustrating how pairwise magnitude and
group-size weighting jointly determine the relationship between $f$ and
mean absolute SMD (Section~2.6), discussed further, with the contrasting
case, in Section~5.3.

After MMWS weighting, this picture reversed for both statistics. The
omnibus $f$ fell to $0.010$ (equivalent $\eta^2=0.0001$), and none of the
three pairwise comparisons excluded zero. The McCaffrey benchmark showed
the same reversal: $\overline{|d|}=0.038$ and $|d|_{\max}=0.057$ after
weighting, both comfortably below the $0.10$ convention. $f$ and
$\overline{|d|}$ agree on the substantive conclusion despite their
differing scales: unambiguous baseline imbalance resolved by weighting
into a covariate with no detectable residual imbalance.

\section{Discussion}

\subsection{Principal Findings}

Four findings emerge. First, $f$ behaved as expected of a balance
diagnostic across the simulated conditions: it rises
with confounding strength under unweighted assignment and separates
correctly specified from misspecified weighting at every $\gamma$ level
tested, though the separation is modest, and the two weighting arms
produced very similar downstream bias (Table~2), suggesting the
interaction-omission misspecification tested here is detectable in balance
before it meaningfully affects estimation. Second, $f$ tracks actual bias
comparably to \citet{mccaffrey2013}'s $\overline{|d|}$, both pooled and
within weighting arm, while $|d|_{\max}$ is consistently the weakest
bias-tracker of the three; neither $f$ nor $\overline{|d|}$ dominates the
other. Third, $f$ and $\overline{|d|}$ are not on the same numeric scale:
their ratio follows an exact, derivable relationship (Section~2.6) rather
than incidental estimation noise. Fourth, this comparative ranking is not
an artifact of testing only a linear outcome model aligned with what these
diagnostics measure: under the nonlinear/interaction robustness scenario
(Section~2.7.6), $f$ and $\overline{|d|}$ remained closely matched and
both ahead of $|d|_{\max}$ in every weighting arm and every $k$ tested,
though all three diagnostics' absolute correlation with bias declined
under weighting specifically, as the nonlinear structure they cannot
directly track came to account for a larger share of the remaining bias.

\subsection{Why the Bias Correlation Is Not Higher}

Within a single weighting arm, $f$'s correlation with actual bias (0.79 for
correctly specified weighting) is positive and precisely estimated at this
sample size, but well below 1. This is expected rather than a shortcoming.
Bias in the outcome estimate reflects the entire data-generating process
--- sampling variability in the outcome regression, the specific realized
covariate draws, and the residual noise term --- of which covariate
imbalance is only one contributing input. A balance diagnostic, however
constructed, describes the treatment-assignment mechanism; it cannot, and
should not be expected to, perfectly predict an estimate built from an
additional, independent source of randomness. The finding to take from the within-arm correlations reported in Section~3
is not that 0.79 is a modest correlation, but that it is a strong
one for a diagnostic computed from an entirely separate model than the one
producing the estimate it is meant to anticipate. The nonlinear/interaction
robustness analysis (Section~2.7.6) reinforces this interpretation: once
weighting removed much of the mean imbalance, correlations fell to
approximately 0.69 because a larger share of downstream bias arose through
outcome structure not represented by mean-balance diagnostics.

\subsection{Practical Advice}

Four recommendations follow directly from the results. First, $f$ should be
reported alongside its per-level and pairwise decomposition
(Sections~2.4--2.5), not as a single omnibus number in isolation: a
root-mean-square statistic computed across $k$ groups is, by construction,
capable of masking a single severely imbalanced pair among several
well-balanced ones, the same concern raised for the $k$-sample
distributional tests in \citet{linden2026ksample}.
Second, $f$ and $\overline{|d|}$ should not be judged against the same
numeric threshold. An investigator accustomed to Austin's \citep{austin2009}
convention that an SMD near $0.1$ indicates negligible imbalance should not
apply that threshold directly to $f$: Table~3 shows $f$ running
systematically below $\overline{|d|}$, by a margin that is itself a
function of $k$ and group-size balance rather than a fixed constant.
Cohen's \citep{cohen1988} own small/medium/large benchmarks (0.10/0.25/0.40)
were developed for interpreting effect sizes in power-analysis and
hypothesis-testing contexts, a different inferential purpose from judging
whether residual imbalance after adjustment makes residual confounding
acceptably unlikely, and we do not treat those benchmarks as validated for
the latter use here. This paper does not establish a validated balance
threshold for $f$; an investigator wanting one should derive it directly
from the relationship between $f$ and downstream bias in a design resembling
their own, following the approach of Section~3, rather than import a
threshold calibrated for a different purpose. The direction of any
group-size-driven divergence between $f$ and $\overline{|d|}$ is not fixed
in general --- our own results show it changing sign as confounding
strength increases (Table~3) --- so no universal rule of thumb (e.g., ``$f$
understates imbalance under unequal groups'') should be assumed without
checking the specific design at hand.

Third, although the two statistics tracked bias comparably here, they are
not mechanically identical. Because $f$ combines the pairwise $d_{jj'}$'s
quadratically but also weights each pair by $w_jw_{j'}$ (Section~2.6), it
gives relatively greater influence than a linear average to large pairwise
imbalances, holding pair weights comparable. This advantage is not
unconditional, however: it can be offset, or reversed, when the largest
imbalance involves small treatment groups, for which $w_jw_{j'}$ is itself
small (point four, below).
$|d|_{\max}$ remains the statistic most directly targeted at detecting a
single worst pair, since it responds to that pair alone rather than any
weighted combination, but $|d|_{\max}$ was also the weakest bias-tracker of
the three statistics tested in every weighting arm (Section~3), consistent
with a maximum carrying more sampling variance across replicates than a
mean- or root-mean-square-based statistic. Whether $f$ or $|d|_{\max}$ is
preferable for detecting an isolated bad pair therefore depends on that
pair's group sizes as much as on the tradeoff between targeting and
sampling variability, and the choice of diagnostic should be reported
transparently \citep{lindenroberts2005}.

Fourth, $f$'s $n_j/N$ weighting is a substantive choice, not merely a
technical detail, and carries a consequence worth making explicit:
imbalance concentrated in a small treatment group contributes relatively
little to the omnibus statistic, while imbalance concentrated in a large
group contributes more. This may have an intuitive justification when the
target estimand weights groups in proportion to their representation in
the target population; observed treatment-arm proportions alone do not
automatically establish that this is the appropriate weighting, since the
target population and estimand, not merely the sample at hand, determine
how much relative importance each treatment contrast should carry. For an
estimand where every pairwise treatment contrast is of equal substantive
interest regardless of group size --- comparing several rare treatment
variants on equal footing, for instance --- $f$'s weighting may be
undesirable, since it can understate imbalance that matters as much to the
investigator as any other. Consider three treatment arms with
prevalences of $80\%$, $15\%$, and $5\%$ in a multivalued-treatment
study: an omnibus statistic that heavily downweights imbalance involving
the $5\%$ arm may be reassuring for an estimand targeting the overall
treated population, while masking a substantively important contrast
involving that less-common treatment. The applied example illustrates
the other side of this mechanism: with $87\%$ of observations in Control
and only $8\%$/$5\%$ in the Calls/RTM arms, the largest unweighted
pairwise SMD (Control vs.\ RTM, $d=-0.201$) involved the \emph{majority}
group rather than two minority arms ($f=0.065$, $\overline{|d|}=0.134$,
Table~5), illustrating how pairwise magnitude and group-size weighting
jointly determine the relationship between $f$ and mean absolute SMD
(Section~2.6). Had the largest
imbalance instead fallen between the two minority arms, as in the
weighting-driven divergence documented more generally in Table~3 and
Section~2.6, $f$ could have understated it substantially relative to
$\overline{|d|}$.

\subsection{Limitations}

Three limitations qualify these findings. First, the nonlinear/interaction
robustness scenario (Section~2.7.6) tested one outcome
structure, adapted from the linear model of \citet{zhaoyang2022}; the
comparative ranking among diagnostics was unchanged under it, but other
nonlinear forms, including \citeauthor{zhaoyang2022}'s own nonlinear
scenario for their balance measure, were not evaluated. Second, only one form of generalized-propensity-score
misspecification was tested --- omission of a single interaction term ---
and the resulting bias difference between correctly specified and
misspecified weighting was modest; a more severe misspecification, such as
omitting a confounder entirely, might separate the two weighting arms, and
the diagnostics that track them, more sharply. Third, $f$ itself carries no
closed-form confidence interval in this implementation. A point estimate alone is
enough to rank candidate models or weighting schemes against one another,
but not to judge how much confidence that ranking deserves: at smaller
$N$, or when $f$ sits close to whatever decision threshold an investigator
has adopted for their own application, a CI is what distinguishes
a value that is reliably below threshold from one that is merely a point
estimate whose sampling variability could as easily place it above.
The same applies when comparing $f$ across two candidate models or
weighting specifications directly, where an apparent difference in the
point estimates alone does not indicate whether that difference is
distinguishable from sampling noise. A noncentral-$F$ construction,
analogous to the noncentral-$t$ intervals used for the pairwise
comparisons, was developed and tested but found to be invalid under the
robust variance estimator that sampling weights typically require, and was
withdrawn rather than reported incorrectly. A bootstrap is recommended in
its place; users requiring a formal interval for the omnibus statistic
specifically should budget for the additional computation this entails.

\subsection{Conclusion}

Cohen's $f$ extends the standardized-mean-difference framework already
standard for two-group covariate balance to $k>2$ treatment groups, with an
exact, derivable relationship to both Cohen's $d$ and the existing
aggregate-SMD practice it is intended to formalize, implemented for
arbitrary weighted and covariate-adjusted models via \texttt{esizereg}. It
performs comparably to existing practice as a bias-tracking diagnostic
while offering a decomposition, and a connection to an established
effect-size framework, that pairwise-SMD aggregation does not. The applied
example illustrated both statistics reaching the same substantive
conclusion on real data despite differing in absolute scale: unambiguous
baseline imbalance resolved by weighting into a covariate with no
detectable residual imbalance. For investigators comparing three or more
treatments in observational or covariate-adjusted studies, $f$ offers a
theoretically grounded omnibus summary of balance, but its group-size
weighting means the per-level and pairwise diagnostics should be retained
and consulted directly whenever a contrast involving a less-common
treatment arm is substantively important, rather than relying on the
omnibus statistic alone.

\bibliographystyle{apalike}
\bibliography{refs}

\clearpage

\begin{table}[htbp]
\centering
\caption{Simulation Study Design}
\small
\begin{tabular}{@{}p{0.28\textwidth}p{0.60\textwidth}@{}}
\toprule
Design factor & Levels / description \\
\midrule
Treatment groups ($k$)      & 3, 4, 6 \\
Sample size ($N$)           & 200, 500, 2000, 10{,}000 \\
Confounding strength ($\gamma$) & 0 (none), 0.25, 0.5, 1 (strong) \\
Group-size balance          & Equal: all $k$ groups share the same baseline assignment probability. Unequal: group-specific intercepts produce a skewed baseline split. \\
Covariates                  & $X_1\sim N(0,1)$; $X_2\sim\text{lognormal}(0,0.5)$; $X_3\sim\text{Bernoulli}(0.4)$ \\
Treatment assignment        & Multinomial logit on $X_1,X_2,X_3$, an $X_1{\times}X_3$ interaction, and $\gamma$ \\
GPS weighting                & Unweighted; correctly specified (includes $X_1{\times}X_3$); misspecified (omits it) \\
Outcome model                & $Y=\mu_T+1.5(X_1+X_2+X_3)+\varepsilon$, $\varepsilon\sim N(0,1)$; true group means $\mu_j$ spaced 1.5 apart \citep{zhaoyang2022} \\
Benchmark                    & Mean and maximum absolute SMD across all $\binom{k}{2}$ pairs \citep{mccaffrey2013}, computed on the same replicate \\
\bottomrule
\end{tabular}

\vspace{4pt}
\parbox{0.9\textwidth}{\footnotesize Note: each of the 12 $(k,N)$ combinations was replicated 1000 times per cell across the 8 $\gamma\times$balance cells (96{,}000 replicates in total). Each replicate was evaluated under all three weighting arms on the identical simulated data, giving a fully paired comparison. All analyses were conducted in Stata (version 18) using the community-contributed Stata command \texttt{esizereg}.}
\end{table}

\clearpage

\begin{table}[htbp]
\centering
\caption{Cohen's $f$, $\overline{|d|}$, and $|d|_{\max}$ by weighting arm and confounding strength}
\small
\begin{tabular}{@{}llcccc@{}}
\toprule
Weighting arm & Statistic & $\gamma=0$ & $\gamma=0.25$ & $\gamma=0.5$ & $\gamma=1$ \\
\midrule
Unweighted & $f$ & 0.062 & 0.091 & 0.132 & 0.209 \\
 & $\overline{|d|}$ & 0.088 & 0.129 & 0.188 & 0.296 \\
 & $|d|_{\max}$ & 0.168 & 0.245 & 0.356 & 0.560 \\
\addlinespace
Correct GPS & $f$ & 0.010 & 0.013 & 0.018 & 0.031 \\
 & $\overline{|d|}$ & 0.014 & 0.018 & 0.025 & 0.043 \\
 & $|d|_{\max}$ & 0.030 & 0.036 & 0.051 & 0.086 \\
\addlinespace
Misspecified GPS & $f$ & 0.010 & 0.013 & 0.019 & 0.034 \\
 & $\overline{|d|}$ & 0.014 & 0.018 & 0.027 & 0.047 \\
 & $|d|_{\max}$ & 0.028 & 0.036 & 0.054 & 0.093 \\
\bottomrule
\end{tabular}
\vspace{4pt}
\parbox{0.9\textwidth}{\footnotesize Note: pooled across $k\in\{3,4,6\}$, $N\in\{200,500,2000,10{,}000\}$, and both balance scenarios; 96{,}000 replicates.}
\end{table}

\clearpage

\begin{table}[htbp]
\centering
\caption{Ratio of $f$ to $\overline{|d|}$ by group count, balance, and confounding strength}
\small
\begin{tabular}{@{}llcccc@{}}
\toprule
$k$ & Balance & $\gamma=0$ & $\gamma=0.25$ & $\gamma=0.5$ & $\gamma=1$ \\
\midrule
3 & Equal   & 0.641 & 0.637 & 0.632 & 0.625 \\
  & Unequal & 0.620 & 0.629 & 0.635 & 0.639 \\
\addlinespace
4 & Equal   & 0.707 & 0.703 & 0.699 & 0.697 \\
  & Unequal & 0.680 & 0.698 & 0.703 & 0.710 \\
\addlinespace
6 & Equal   & 0.772 & 0.771 & 0.772 & 0.785 \\
  & Unequal & 0.728 & 0.749 & 0.761 & 0.785 \\
\bottomrule
\end{tabular}
\vspace{4pt}
\parbox{0.9\textwidth}{\footnotesize Note: $f/\overline{|d|}$, pooled across $N$. Under equal group sizes, the minimum, attained at equally-spaced group means and proven in Section~2.6, is $\sqrt{3(k-1)/(4(k+1))}$ (0.612, 0.671, 0.732 at $k=3,4,6$). Departures reflect pairwise-SMD dispersion beyond the equally-spaced configuration and, under unequal group sizes, whether the largest pairwise imbalances involve numerically larger or smaller groups (Section~2.6).}
\end{table}

\clearpage

\begin{table}[htbp]
\centering
\caption{Correlation with actual estimation bias, linear vs.\ nonlinear/interaction outcome model, by weighting arm}
\small
\begin{tabular}{@{}llccc@{}}
\toprule
Weighting arm & Outcome model & $f$ & $\overline{|d|}$ & $|d|_{\max}$ \\
\midrule
Unweighted        & Linear    & 0.953 & 0.962 & 0.921 \\
Unweighted        & Nonlinear & 0.946 & 0.946 & 0.922 \\
\addlinespace
Correct GPS       & Linear    & 0.728 & 0.714 & 0.718 \\
Correct GPS       & Nonlinear & 0.688 & 0.671 & 0.662 \\
\addlinespace
Misspecified GPS  & Linear    & 0.725 & 0.715 & 0.717 \\
Misspecified GPS  & Nonlinear & 0.687 & 0.673 & 0.664 \\
\bottomrule
\end{tabular}
\vspace{4pt}
\parbox{0.9\textwidth}{\footnotesize Note: both columns restricted to the
identical design ($k\in\{3,4,6\}$, $N\in\{500,2000\}$, 1000 replicates per
cell, 48{,}000 replicates per outcome model); the linear-outcome figures
are the by-weighting-arm correlations from the full simulation (Table~2),
recomputed on this $N$ range specifically rather than pooled across all
four $N$ values tested there.}
\end{table}

\clearpage

\begin{table}[htbp]
\centering
\caption{Applied example: predictive model score balance before and after MMWS weighting}
\small
\begin{tabular}{@{}lcc@{}}
\toprule
 & Unweighted & MMWS-weighted \\
\midrule
Cohen's $f$ (omnibus)     & 0.065 & 0.010 \\
Equivalent $\eta^2$       & 0.004 & 0.0001 \\
$\overline{|d|}$ (benchmark) & 0.134 & 0.038 \\
$|d|_{\max}$ (benchmark)  & 0.201 & 0.057 \\
\addlinespace
\multicolumn{3}{l}{Pairwise $d$ [95\% CI]} \\
\quad Control vs.\ Calls  & $-0.189$ [$-0.285$, $-0.094$] & $-0.015$ [$-0.110$, $0.081$] \\
\quad Control vs.\ RTM    & $-0.201$ [$-0.316$, $-0.086$] & $\phantom{-}0.042$ [$-0.073$, $0.157$] \\
\quad Calls vs.\ RTM      & $-0.012$ [$-0.156$, $0.133$]  & $\phantom{-}0.057$ [$-0.088$, $0.201$] \\
\bottomrule
\end{tabular}
\vspace{4pt}
\parbox{0.9\textwidth}{\footnotesize Note: $N=5907$ (5141 Control, 459 Calls, 307 RTM). $\overline{|d|}$ and $|d|_{\max}$ computed from the same three pairwise $d$'s reported below.}
\end{table}

\end{document}